\documentclass[journal, onecolumn]{IEEEtran}
\ifCLASSINFOpdf
  \usepackage[pdftex]{graphicx}
\else
\fi
\usepackage[cmex10]{amsmath}
\usepackage{amssymb}
\usepackage{algorithm}
\usepackage{algpseudocode}
\usepackage{pseudocode}
\newcommand{\subvec}[1]{\underline{S}(#1)}
\newcommand{\supvec}[1]{\overline{S}(#1)}
\newcommand{\scalprod}[2]{\left\langle #1,#2 \right\rangle}
\newcommand{\vectorize}[1]{\overline{#1}}

\newtheorem{theorem}{Theorem}[section]

\newtheorem{proposition}[theorem]{Proposition}
\newtheorem{example}[theorem]{Example}

\newtheorem{remark}[theorem]{Remark}

\begin{document}
%
\title{Nonlinearity Computation for Sparse Boolean Functions}
%
%
%

\author{\c{C}a\u{g}da\c{s}~\c{C}al\i k,~\IEEEmembership{Member,~IEEE}
\thanks{The author is with the Institute of Applied Mathematics, Middle East Technical University, Ankara, Turkey.}%
\thanks{The author is partially supported by T\"{U}B\.{I}TAK under grant no. 109T672.}%
\thanks{This article is a revised version of a part of the author's Ph.D. thesis \cite{Ca13}.}
}
\maketitle

\begin{abstract}
An algorithm for computing the nonlinearity of a Boolean function
from its algebraic normal form (ANF) is proposed. By generalizing
the expression of the weight of a Boolean function in terms of its
ANF coefficients, a formulation of the distances to linear
functions is obtained. The special structure of these distances can be exploited to reduce 
the task of nonlinearity computation to solving an associated binary integer programming problem. The proposed algorithm can be used in cases where applying the Fast Walsh transform is infeasible, typically when the number of input variables exceeds 40.

\end{abstract}


%
\IEEEpeerreviewmaketitle

\section{Introduction}
%
%
%
%

Boolean functions have various applications
in cryptology, coding theory, digital circuit theory, etc. \cite{MS77}. Among the properties associated with a Boolean function, nonlinearity is an important criterion regarding the
security point of view. Nonlinearity is defined as the minimum
distance of a Boolean function to the set of affine functions and
functions used in secrecy systems are expected to have high
nonlinearity in order to resist certain cryptanalytic attacks. This makes
the nonlinearity computation a necessary task in order to prove
that the claimed level of security is achieved.

Boolean functions have several representations and for each
representation, either there is a direct method to compute the
nonlinearity (e.g. Walsh spectrum), or a transformation to a
representation of this type is available. Two common ways of
representing Boolean functions are truth table and algebraic normal
form (ANF). Truth table is the list of all output values of a
Boolean function in a predetermined order. Computing nonlinearity
from the truth table can be done by Fast Walsh transform (FWT),
which constructs the Walsh spectrum of the function and the entry with the maximum absolute value in the spectrum determines nonlinearity.
As the truth table of an $n$-variable Boolean function consists of
$2^n$ entries, the cost of storing the truth table increases
exponentially in $n$. ANF is a preferred representation either
when it is impractical to store the truth table, or it is more
efficient and/or secure to calculate the output of the Boolean function from
an expression of the input variables.

Unless the ANF of a Boolean function belongs to a class that reveals
its nonlinearity (e.g. affine functions), the task of computing the
nonlinearity from the ANF can be performed by constructing the
truth table from the ANF by Fast M\"{o}bius transform and then
applying FWT on it. This process has a computational complexity of
$\mathcal{O}(n2^n)$ both to transform the ANF to truth table and to
apply FWT. Clearly, when $n$ gets larger, say $n>40$, considering
the computation and memory resources of current computers, applying
these transformations becomes infeasible.

Expression of the weight of a Boolean function in terms of its ANF
coefficients is introduced by Carlet and Guillot \cite{CG99}, which
allows one to compute the weight from the ANF with a
complexity of $\mathcal{O}(2^p)$ operations, if the Boolean function
consists of $p$ monomials. Two related works utilizing this expression
propose more efficient methods for Walsh coefficient computation and
weight computation from the ANF \cite{GS09, CD12}. In this work, by investigating the
distance of a Boolean function to the set of affine functions in
terms of ANF coefficients, an algorithm to compute the nonlinearity
for Boolean functions with high number of inputs is devised.

This paper is organized as follows: Section 2 gives basic
definitions about Boolean functions and introduces the notation used
in the paper. Section 3 investigates the expression of the distance
to the set of affine functions in terms of ANF coefficients by introducing
the Linear Distance Matrix. In Section 4, the algorithm for
computing the nonlinearity by exploiting the structure of this matrix is described, with a discussion on the complexity of the algorithm, techniques for performance improvement and implementation results. Section 5 gives the conclusion.

\section{Preliminaries}


Let $\mathbb{F}_2 = \{0,1\}$ be the finite field with two elements
and $\mathbb{F}_2^n$ be the $n$-dimensional vector space over
$\mathbb{F}_2$. The addition operation in $\mathbb{F}_2$ and
$\mathbb{F}_2^n$ will be denoted by $\oplus$, whereas $+$ will be
used for integer addition. Binary logical operators \texttt{and} and
\texttt{or} will be denoted by $\wedge$ and $\vee$, respectively.
Unary operator $\neg$ will stand for the bitwise complement of a
vector. Inner product of two vectors $x,y\in\mathbb{F}_2^n$ is
$\scalprod{x}{y}=x_1y_1\oplus\cdots\oplus x_ny_n$. Support of a
vector $x = (x_1, \ldots, x_n) \in \mathbb{F}_2^n$ is defined as
$supp(x)=\{i \mid x_i\neq 0\}$. Weight $wt(x)$ of a vector is the
number of its non-zero components, i.e., $wt(x)=|supp(x)|$.

Elements of $\mathbb{F}_2^n$ can be identified with integers modulo
$2^n$ by associating the vector $\vectorize{x} = (x_1, \ldots, x_n) \in
\mathbb{F}_2^n$ with the integer $x =
\sum\limits_{i=1}^{n}{x_i2^{n-i}}$.

An $n$-variable Boolean function $f : \mathbb{F}_2^n \rightarrow
\mathbb{F}_2$ specifies a mapping from the $n$-dimensional vector
space to $\mathbb{F}_2$. The sequence $T_f = ( f(\vectorize{0}),
f(\vectorize{1}), \ldots , f(\vectorize{2^n-1}) )$ is called the
\emph{truth table} of $f$. The \emph{support} (resp. \emph{weight})
of $f$ is the support (resp. weight) of its truth table, hence the
weight of the function is $wt(f)=|supp(f)|=|\{x \in \mathbb{F}_2^n
\mid f(x)=1\}|$. $f$ is called \emph{balanced} if $wt(f)=2^{n-1}$. A
Boolean function can be represented as a multivariate polynomial
called the \textit{algebraic normal form (ANF)} such that

\begin{equation}
\label{eqn_anf_vec} f(x_1,\ldots,x_n)=\bigoplus\limits_{u\in
\mathbb{F}_2^n}{a_ux^u}
\end{equation}
where $x^u=x_1^{u_1}x_2^{u_2}\ldots x_n^{u_n}$ is a \emph{monomial}
composed of the variables for which $u_i=1$ and $a_u \in
\mathbb{F}_2$ is called the ANF coefficient of $x^u$. Degree of the
monomial $x^u$ is $wt(u)$ and the highest degree monomial with the non-zero
ANF coefficient determines the degree of a function.

Throughout the text, for the sake of simplicity, the ANF
coefficients of a Boolean function will be denoted by $a_i$ for $i\in\{0,1,\cdots,2^n-1\}$ being an integer. In this notation, $a_i$
is the coefficient of the monomial $x^{\vectorize{i}}$. For example,
$a_0,a_1,a_2,a_3$ and $a_{2^n-1}$ are the ANF coefficients of the
monomials $1,x_n,x_{n-1},x_{n-1}x_{n-2}$ and $x_1x_2\cdots x_n$,
respectively.

Distance between two functions is measured with the number of truth
table entries they differ by, i.e.,
\[d(f,g)=|\{x\in\mathbb{F}_2^n\mid f(x)\neq g(x)\}|=wt(f\oplus g).\]

For $w\in \mathbb{F}_2^n$ and $c\in \mathbb{F}_2$, Boolean functions
of the form
\[
f(x)= \scalprod{w}{x}\oplus c
\]
are called \emph{affine} functions, and in particular when $c=0$
they are called \emph{linear} functions. $l_w$ and
$l_w^{'}=l_w\oplus 1$ will be used to denote linear functions and
their complements, respectively. The set of affine functions
is shown by $A_n$. Nonlinearity $N_f$ of a Boolean function $f$
is the minimum of the distances between $f$ and the set of affine
functions, i.e., \[N_f= \min\limits_{g\in A_n}{d(f,g)}.\]

For two vectors $x,y \in \mathbb{F}_2^n$, if $supp(x) \subseteq
supp(y)$, $x$ is called a sub-vector of $y$, and $y$ is called a
super-vector of $x$, which will be denoted by $x\preceq y$ or $y
\succeq x$. The set of sub-vectors of a vector is denoted by
$\underline{S}(x)=\{y \in \mathbb{F}_2^n \mid y \preceq x\}$, and
the set of super-vectors of a vector is denoted by
$\overline{S}(x)=\{y \in \mathbb{F}_2^n \mid y \succeq x\}$. The
following observations follow from the definition of sub-vector and
super-vector:

\begin{eqnarray}
|\subvec{x}| &=& 2^{wt(x)} \label{eqn_subvec}\\
|\supvec{x}| &=& 2^{n-wt(x)}\label{eqn_supvec}
\end{eqnarray}

\section{Distance to Linear Functions}

Any integer valued function $G: \mathbb{F}_2^m\rightarrow
\mathrm{Z}$ defined on binary $m$-tuples can be represented by

\begin{equation}
\label{eqn_general}
G(x_1,\ldots,x_m)=\sum\limits_{I\subseteq\{1,\cdots,m\}}^{}{\lambda_I
x^I}
\end{equation}
where $x_i\in\mathbb{F}_2$ for $1\leq i \leq m$ and $\lambda_I \in
\mathrm{Z}$ is the coefficient of the product $x^I=x_{i_1}\ldots
x_{i_d} \mbox{ for } I=\{i_1,\cdots,i_d\}.$ Here, $x_i\in
\mathbb{F}_2$ implies $x_i^2=x_i$, therefore all the terms $x^I$ are distinct products of input variables and the function
has $2^m$ terms. For the functions of interest to this study, $G$ maps the ANF
coefficients of an $n$-variable Boolean function to the distance to a particular linear function,
hence the number of input variables is $m=2^n$. When $m$ is large, it
is not possible to list all of the ANF coefficients of a Boolean
function. Instead, the support of the ANF is supplied, which is assumed to be relatively small sized. For an input
$x=(x_1,\ldots,x_m)$ and $supp(x)=\{i_1,\cdots,i_p\}$, the output of
the function will be the sum
\[
G(x)=\sum\limits_{I\subseteq\{i_1,\cdots,i_p\}}\lambda_I
\]
consisting of $2^p$ coefficients $\lambda_I$, associated with all nonzero products $x^I$. 
Using this approach, the function $G$ can be evaluated in $\mathcal{O}(2^p)$ operations if the complexity of computing each $\lambda_I$ is negligible. In the following part of this section, it
will be shown that the coefficients of the functions mapping the ANF
coefficients to the weight and to the distance to a particular
linear function can be computed in a very simple way. On the
contrary, this is not the case for nonlinearity. This will lead to
constructing a method to compute the nonlinearity without trying to
compute the coefficients of the function, but by exploiting the
properties of the previously mentioned functions' coefficients that
could be easily computable.

The function mapping the ANF coefficients to the weight of a Boolean
function was introduced in \cite{CG99}, which will be called the \emph{weight function}.
In order to render the remaining of the
text more comprehensible, derivation of the coefficients of this function is going to be explained.

The output of a Boolean function can be calculated in terms of its ANF coefficients as follows:

\begin{equation}
\label{eqn_ttanf}
f(x)= \bigoplus\limits_{u\preceq x}^{}a_u.
\end{equation}
Table \ref{tab_anf2tt} shows the truth table entries of 3-variable
Boolean functions in terms of ANF coefficients. When an ANF
coefficient $a_i$ contributes to the output of a Boolean function at
a point, it is said that $a_i$ appears in that truth table entry.
The sum of the truth table entries gives the weight of the
function. Weight can be expressed as a function over the integers by
replacing the addition operation in $\mathbb{F}_2$ in each truth
table entry with integer addition operation. $\mathbb{F}_2$ addition
can be converted to addition over the integers by a well known
formula, generalizing the fact that $a\oplus b=a+b-2ab$:

\begin{eqnarray}
\label{eqn_poincare}
\bigoplus\limits_{i=1}^{m}{a_i}&=&\sum\limits_{k=1}^{m}{(-2)^{k-1}}\sum\limits_{1\leq
i_1 < \cdots < i_k \leq m}{a_{i_1}\cdots a_{i_k}}.
\end{eqnarray}

This formula states that the expression of addition over the
integers consists of all combinations of products of terms, with a leading
coefficient related to the number of terms in the product, such as all
degree two terms having the coefficient $-2$ and all degree three terms
having the coefficient $4$, etc.

\begin{proposition}
\label{prop_anfcount} For an $n$-variable Boolean function, ANF
coefficient $a_u$ contributes to the output in $2^{n-wt(u)}$ points.
\end{proposition}
\begin{proof}
From (\ref{eqn_ttanf}), $a_u$ contributes to the output of the
function at a point $x$ if  $u \preceq x$. This means that $a_u$
contributes to the function output at points $\supvec{u}$, whose
size is equivalent to $2^{n-wt(u)}$ by (\ref{eqn_supvec}).
\end{proof}

\begin{proposition}
\label{prop_union} For each set $A=\{a_{u_1},\cdots,a_{u_k}\}$ of
ANF coefficients, there exists a coefficient $a_v$ with the property
$supp(v)=\bigcup\limits_{1\leq i \leq k}{supp({u_i})}$, such that
the truth table entries $a_v$ appears are exactly the same as the
truth table entries all the coefficients in $A$ appear together. Such a coefficient $a_v$ will
be called the representative coefficient of $\prod\limits_{1\leq i \leq k}{a_{u_i}}$.
\end{proposition}
\begin{proof}
From (\ref{eqn_ttanf}), each $a_{u_i}$ appears in truth table
entries $\supvec{u_i}$. If $y$ is a truth table entry such that all
$a_{u_i}$'s appear together then $\bigcup\limits_{1\leq i \leq
k}^{}{supp(u_i)}\subseteq supp(y)$. Call the minimum weight vector
satisfying this property $v$, the case where the set equality
occurs. Then, by definition of a super-vector, the set $\supvec{v}$
also has the same property of covering the supports of $u_i$'s, and
these are the points the ANF coefficient $a_v$ appears in the truth table.
Therefore, if vector $v$ is chosen such that
$supp(v)=\bigcup\limits_{1\leq i \leq k}^{}{supp(u_i)}$, $a_v$
appears at exactly the same truth table entries as
$\{a_{u_1},\cdots,a_{u_k}\}$ appear together.
\end{proof}

Combining Proposition \ref{prop_anfcount} and (\ref{eqn_poincare}),
one obtains the weight function of an $n$-variable Boolean function
as follows:

\begin{equation}\label{anf2nl_eqn_weight}
F(a_0,\dots,a_{2^n-1})=\sum\limits_{I\subseteq\{0,\dots,2^n-1\}}{\lambda_I a^I}
\end{equation}
where $\lambda_I\in\mathrm{Z}$ is called the \emph{weight
coefficient} of the product $a^I=\prod\limits_{i\in I}^{}{a_i}$. If
$I=\emptyset$, the value of $\lambda_I$ is found to be zero. For a
non-empty set $I=\{i_1,\cdots,i_k\}$, the value of $\lambda_I$ is
determined by two factors:
\begin{eqnarray}
\label{eqn_lambda} \lambda_I&=&d_I n_I,\\
\label{eqn_di} d_I&=&(-2)^{k-1},\\
\label{eqn_ni} n_I&=&2^{n-wt(\vectorize{i_1}\vee \dots \vee
\vectorize{i_k})}.
\end{eqnarray}
The value of $d_I$ comes from (\ref{eqn_poincare}) and $n_I$ is the
number of times the product $a^I$ occurs in the truth
table entries when each entry is expressed with addition over the integers. The value of $n_I$ is calculated by Proposition \ref{prop_union}, which is related to the number of distinct input variables appearing in the monomials of contributing ANF coefficients.

\begin{example}
Let $n=3$ and consider the weight coefficient of $a_1a_2$ in the weight
function. After the conversion of binary addition to integer
addition, $a_1a_2$ appears in a truth table entry if and only if
both of $a_1$ and $a_2$ are present at that entry. According to
Proposition \ref{prop_union}, $a_3$ appears as many times as
$a_1a_2$ appears in the truth table since $\vectorize{1} \vee
\vectorize{2}=\vectorize{3}$, and the entries they appear are
exactly the same, which are the fourth and the last rows of Table \ref{tab_anf2tt}. Hence, by Proposition
\ref{prop_anfcount}, $a_1a_2$ appears $2^{3-wt(\vectorize{3})}=2$ times and since $a_1a_2$ consists of two
terms, each one of these terms will have the constant
$(-2)^{2-1}=-2$ arising from the conversion of addition
(\ref{eqn_poincare}). As a result, the coefficient of $a_1a_2$
becomes $(-2).2=-4$.
\end{example}

\begin{table}[!t]
\renewcommand{\arraystretch}{1.3}
\caption{Truth Table in terms of ANF Coefficients.}
\label{tab_anf2tt}
\centering
\begin{tabular}{ccc|c|l}
\hline
$x_1$& $x_2$& $x_3$ & ANF & Truth table\\
\hline
0& 0&0& $a_0$ & $a_0$\\
0& 0&1& $a_1$ & $a_0 \oplus a_1$\\
0& 1&0& $a_2$ & $a_0 \oplus a_2$\\
0& 1&1& $a_3$ & $a_0 \oplus a_1 \oplus a_2 \oplus a_3$\\
1& 0&0& $a_4$ & $a_0 \oplus a_4$\\
1& 0&1& $a_5$ & $a_0 \oplus a_1 \oplus a_4 \oplus a_5$\\
1& 1&0& $a_6$ & $a_0 \oplus a_2 \oplus a_4 \oplus a_6$\\
1& 1&1& $a_7$ & $a_0 \oplus a_1 \oplus a_2 \oplus a_3 \oplus a_4 \oplus a_5 \oplus a_6 \oplus a_7$\\
\hline
\end{tabular}
\end{table}

\subsection{The Linear Distance Matrix}

Now, the method of obtaining the coefficients of the function which
outputs the distance of a Boolean function to a linear function, in
terms of ANF coefficients will be described. Essentially, the
distance between a Boolean function $f$ and a linear function $l_w$
is equivalent to $wt(f\oplus l_w)$. So, an investigation of how the
coefficients in the weight function change when the truth table of $f$ is
merged with a linear function is necessary. The weight function when
computed with new coefficients will be called the \emph{distance
function}, producing the output $d(f, l_w)$. This is a generalization of the weight function
defined as

\begin{equation}\label{anf2nl_eqn_weight2}
F_w(a_0,\dots,a_{2^n-1})=\sum\limits_{I\subseteq\{0,\dots,2^n-1\}}{\lambda_I^w
a^I}
\end{equation}
where $w\in \mathbb{F}_2^n$ specifies which linear function the
distance is measured, $\lambda_I^w\in \mathrm{Z}$ is the \textit{distance coefficient} of the product $a^I$. When
$w=\vectorize{0}$, (\ref{anf2nl_eqn_weight2}) is equivalent to
(\ref{anf2nl_eqn_weight}) and outputs the weight. The following
proposition states how the coefficients $\lambda_I^w$ are obtained.

\begin{proposition}
\label{prop_lambdaw}
Let $\lambda_I$ be the weight coefficient of $a^I$ in the weight function of an $n$-variable Boolean function $f$ and let $a_v$ be the representative ANF coefficient of $a^I$. The distance coefficient $\lambda_I^{w}$ of $a^I$ in the
distance function $F_w$ for a nonzero $w\in\mathbb{F}_2^n$ is
\begin{eqnarray*}
\label{eqn_newci}
 \lambda_I^{w} =
\left\{\begin{array}{lrc} 2^{n-1}, & \mbox{ if } I=\emptyset,\\
\phantom{-}\lambda_I, &\mbox{ if } I\neq\emptyset \mbox{ and } w \preceq v \mbox{ and } wt(w) \mbox{ is even},\\
-\lambda_I, &\mbox{ if } I\neq\emptyset \mbox{ and } w \preceq v \mbox{ and } wt(w) \mbox{ is odd},\\
\phantom{-\lambda}0,&\mbox{ otherwise}.
\end{array}
\right.
\end{eqnarray*}
\end{proposition}

\begin{proof}
Adding a nonzero linear function $l_w$ to $f$ complements the truth
table of $f$ at $2^{n-1}$ points. At these points, the new value of
the function becomes $1\oplus f(x)$, which is equivalent to $1-f(x)$
in integer arithmetic. Considering the integer valued expression of
the truth table entries in terms of ANF coefficients, this
corresponds to negating the terms and producing a constant value of
$1$ that is independent of the ANF coefficients. This proves
$\lambda_I^w=2^{n-1}$ for $I=\emptyset$.

In order to find out the values of other coefficients, at how many
points $supp(l_w)$ coincides with the truth table entries the
product $a^I$ appears must be calculated. If all (resp. half, none)
of the points where $a^I$ appears in the truth table coincide with
$supp(l_w)$, then $\lambda_I^w$ will be $-\lambda_I$ (resp. $0$,
$\lambda_I$).

Linear function $l_w$ identified by the vector $w\in\mathbb{F}_2^n$
has $supp(w)$ terms and takes on the value $1$ whenever an odd sized
combination of its terms are added. Namely,
\begin{equation}
\label{eqn_lwsupport}
 supp(l_w)=\{x\in\mathbb{F}_2^n \mid \, \#\{supp(x)\cap supp(w)\} \mbox{ is odd.}\}.
\end{equation}
This also means that, if $x\in supp(l_w)$ then $supp(x)=I\cup J$
such that $I\subseteq supp(w)$ with $|I|\equiv 1\pmod{2}$ and
$J\subseteq \{1,\ldots, n\} \setminus supp(w)$, i.e., the components
which are not in the support of $w$ can be chosen freely since they
do not contribute to the output of $l_w$. On the other hand, if
$a_v$ is the representative ANF coefficient of the product $a^I$, the truth
table entries where $a^I$ appears is $\supvec{v}$ by
(\ref{eqn_ttanf}).
\begin{itemize}
\item Assume $w \preceq v$. For a vector $x\in\mathbb{F}_2^n$ to be both in $supp(l_w)$ and $\supvec{v}$, $supp(w)\subseteq supp(x)$ is necessary. Otherwise, if any component $j\in supp(w)$ of $x$ is taken to be zero, $x$ will not be in $\supvec{v}$,
because $j\in supp(w)$ and $w \preceq v$ implies $j\in supp(v)$, which means the $j^{th}$ component will always be 1 in $\supvec{v}$. Hence, intersection occurs at the
points $\supvec{w}$, i.e., all the terms in $l_w$ must be chosen. If
$wt(w)$ is even, the linear function $l_w$ gets the value zero at
these points and the intersection becomes the empty
set, proving $\lambda_I^w=\lambda_I$. Following the same argument,
if $wt(w)$ is odd and all the terms in $l_w$ are chosen, $l_w$ attains the value 1 at the points $\supvec{w}$.
Since $w \preceq v$ implies $\supvec{v}\subseteq\supvec{w}$, all the points the term $a_v$ appears in the truth table coincide with $supp(l_w)$ and the terms at these points will be negated due to complementation, leading to $\lambda_I^w=-\lambda_I$.

\item Assume $w \npreceq v$. This implies $A=supp(w) \setminus
supp(v)\neq\emptyset$. Let $x\in\mathbb{F}_2^n$ such that $supp(x)=supp(w)\cap supp(v)$. Then
it is easy to show that half of the vectors in $\supvec{x}$ have even weight and half of them have odd weight. Because for any
$y_1\in\supvec{x}$ with $|supp(y_1)\cap A| \equiv 1\pmod{2}$ a corresponding vector $y_2$ such that $|supp(y_2)\cap A|\equiv 0\pmod{2}$ can be found. A consequence of this is the output of the linear function $l_w$ at points $y_1$ and $y_2$ are complements of each other. This means that half of the vectors in $\supvec{v}$ are also in $supp(l_w)$ and half of them are not. Therefore, in the summation of truth table entries at positions $\supvec{v}$, terms cancel each other making $\lambda_I^w=0$.
\end{itemize}
\end{proof}

Now consider all $2^n$ functions $F_w$ which map the ANF
coefficients of a Boolean function $f$ to the distance $d(f, l_w)$,
with the attention being on the distance coefficients of $a^I$
with $|I|=1$, i.e., distance coefficients of $a_i$ for
$i\in\{0,\cdots,2^n-1\}$. According to (\ref{prop_union}), for any
product $a^I$ with $|I|>1$, a representative coefficient from this
set can be used. The $2^n\times 2^n$ matrix whose $i^{th}$ row
consists of such coefficients $\lambda_I^{\vectorize{i}}$ will be
called the \emph{Linear Distance Matrix (LDM)} of order $n$, denoted by
$M^n$. In view of \eqref{eqn_lambda} and Proposition \ref{prop_lambdaw}, each entry of LDM can be defined as follows:
\begin{equation}
\label{eqn_ldm2}
 M_{i,j}^n = \left\{\begin{array}{lrc}
(-1)^{wt(\vectorize{i})}2^{n-wt(\vectorize{j})}, &\mbox{ if } \vectorize{i} \preceq \vectorize{j},\\
0,&\mbox{ otherwise.}
\end{array}
\right.
\end{equation}

\begin{table}[!t]
\renewcommand{\arraystretch}{1.3}
\caption{Linear Distance Matrix for $n=3$.}
\label{tab_ldm3}
\centering
\begin{tabular}{l|rrrrrrrr}
\hline
&$a_0$ & $a_1$ & $a_2$ & $a_3$ & $a_4$ & $a_5$ & $a_6$ & $a_7$\\
\hline
$l_0$&8   &   4   &   4   &   2   &   4   &   2   &   2   &   1       \\
$l_1$&0   &   -4  &   0   &   -2  &   0   &   -2  &   0   &   -1      \\
$l_2$&0   &   0   &   -4  &   -2  &   0   &   0   &   -2  &   -1      \\
$l_3$&0   &   0   &   0   &   2   &   0   &   0   &   0   &   1       \\
$l_4$&0   &   0   &   0   &   0   &   -4  &   -2  &   -2  &   -1      \\
$l_5$&0   &   0   &   0   &   0   &   0   &   2   &   0   &   1       \\
$l_6$&0   &   0   &   0   &   0   &   0   &   0   &   2   &   1       \\
$l_7$&0   &   0   &   0   &   0   &   0   &   0   &   0   &   -1      \\
\hline
\end{tabular}
\end{table}
Table \ref{tab_ldm3} shows the LDM of order $3$. The entries of
$n^{th}$ order LDM are closely related to the Sylvester-Hadamard
matrix $H_n$, which is defined as
\begin{eqnarray}
\label{eqn_hadamard}
H_0&=&\left[1\right],\\
H_n&=&\begin{bmatrix}
1&\phantom{-}1\\
1&-1%
\end{bmatrix}\otimes  H_{n-1}, \mbox{ for } n\geq 1,
\end{eqnarray}
where $\otimes$ refers to the Kronecker product of matrices. Each
row (or column) of $H_n$ represents the truth table of a linear
function, whose entries are transformed from $(0,1)$ to $(1,-1)$. Table
\ref{tab_h3} shows $H_3$ where only the signs of the entries are
shown. '+' and '-' denote the points the function takes on the
values $0$ and $1$ respectively.
\begin{proposition}
\label{prop_ldmhadamard}
$M_{i,j}^n$ can be obtained by adding the entries of the $i^{th}$ row of
$H_n$ at columns $\supvec{\vectorize{j}}$:
\end{proposition}
\begin{proof}
Since the $i^{th}$ row of $H_n$ represents the truth table of
$l_{\vectorize{i}}$, the distance coefficient of $a_v$ in the distance
function can be calculated by adding the '+' and '-' values of $H_n$
in the $i^{th}$ row at columns $\supvec{v}$. This gives how
many times the sign of $a_v$ will be positive and negative in the integer valued expression of truth table entries when the
linear function $l_{\vectorize{i}}$ is added to a Boolean function. This sum corresponds to the distance coefficient of $a_v$.
\end{proof}

LDM can also be expressed with the following recursive structure:
\begin{eqnarray}
\label{eqn_ldm}
M^{n,0}&=&\left[1\right],\\
M^{n,i}&=&\begin{bmatrix}
2&\phantom{-}1\\
0&-1%
\end{bmatrix}\otimes  M^{n,i-1}, \mbox{ for } 1\leq i \leq n,\\
M^{n,n}&=&M^n.
\end{eqnarray}
This recursive structure can be explained with Sylvester-Hadamard
matrices. As stated in Proposition \ref{prop_ldmhadamard}, the entries of $M^n$ correspond to the sum of particular entries of $H_n$. When the dimension is extended from $n$ to $n+1$,
$H_n$ grows according to (\ref{eqn_hadamard}). As a result of
the way $H_n$ is duplicated to produce $H_{n+1}$, the values of
$M^n$ are doubled for ANF coefficients that do not contain the newly
introduced variable, which corresponds to the upper left quarter of
$M^{n+1}$. The upper right quarter of $M^{n+1}$ will be equal to
$M_n$ as this part of $H_{n+1}$ is equal to $H_n$. The lower right
quarter will be $-M^n$ since the entries of $H_{n+1}$ at this part have opposite signs with $H_n$, and the lower left quarter will be the
matrix consisting of all zeros because for each $\supvec{v}$, half of the entries will be
positive and the other half will be negative.

\begin{table}[!t]
\renewcommand{\arraystretch}{1.3}
\caption{$H_3$: Sylvester-Hadamard Matrix of order three.}
\label{tab_h3} \centering
\begin{tabular}{l|c c c c| c c c c}
\hline & $a_0$ & $a_1$ &$a_2$ &$a_3$ &$a_4$ &$a_5$ &$a_6$ &$a_7$\\
\hline
$l_0$&+&+&+&+&+&+&+&+\\
$l_1$&+&-&+&-&+&-&+&-\\
$l_2$&+&+&-&-&+&+&-&-\\
$l_3$&+&-&-&+&+&-&-&+\\ \hline
$l_4$&+&+&+&+&-&-&-&-\\
$l_5$&+&-&+&-&-&+&-&+\\
$l_6$&+&+&-&-&-&-&+&+\\
$l_7$&+&-&-&+&-&+&+&-\\
\hline
\end{tabular}
\end{table}

The first row of the LDM contains the distance coefficients for calculating the distance
to the linear function $l_0=0$, which also corresponds to the
weight. The coefficients in this row are also equivalent to the weight coefficients of the weight function. The second and third rows contain the distance coefficients for calculating the
distances to the linear functions $l_1=x_n$ and $l_2=x_{n-1}$ respectively, and so on. The first row of the LDM will often be an exceptional case for the rest of the
discussions in this chapter because once the weight of the function is
calculated in the first step of the nonlinearity computation
algorithm which is going to be explained in the next section, this row will no longer be needed. Some properties of the LDM derived from (\ref{eqn_ldm2}) are as follows:
\begin{remark}
\label{rmk_ldm1} Entries in the $j^{th}$ column take on the values
from the set $\{0, \pm 2^k\}$, where $k=n-wt(\vectorize{j})$.
\end{remark}

Except the first and the last columns of the LDM, all three values
mentioned in Remark \ref{rmk_ldm1} appear in a column. In the first
column, there is only one nonzero entry which is positive and all
the other entries are zero. In the last column, there is no entry
with a zero value.

\begin{remark}
\label{rmk_ldm4} Nonzero entries of the $j^{th}$ column are at
positions $x$ where $\vectorize{x}=\subvec{\vectorize{j}}$.
\end{remark}

\begin{remark}
\label{rmk_ldm2_0} Nonzero entries of the $i^{th}$ row are at
positions $x$ where $\vectorize{x}=\supvec{\vectorize{i}}$.
\end{remark}

\begin{remark}
\label{rmk_ldm2} Nonzero entries of the $i^{th}$ row are
positive if $wt(\vectorize{i})$ is even, and negative if $wt(\vectorize{i})$ is odd.
\end{remark}

\begin{remark}
\label{rmk_ldm3} Let $j_1$ and $j_2$ be two column indices in the LDM.
Then the following holds:
\begin{itemize}
\item If $\vectorize{j_1} \preceq \vectorize{j_2}$ then $M_{i,j_1} \neq 0$
implies $M_{i,j_2} \neq 0$.
\item If $supp(\vectorize{j_1})\cap supp(\vectorize{j_2})=\emptyset$ then at least one of
$M_{i,j_1}$ and $M_{i,j_2}$ is zero, except for $i=0$.
\end{itemize}
\end{remark}

\begin{remark}
\label{rmk_ldm5} Zero entries of the $j^{th}$ column are at positions $x$
where $supp(\vectorize{x})\cap (\{1,\cdots,n\}\setminus
supp(\vectorize{j}))\neq \emptyset$.
\end{remark}

\begin{proposition}
\label{prop_ldmpositive} If the $i^{th}$ row of the LDM is used to
measure the distance to the linear function $l_{\vectorize{i}}$,
negating each entry of the $i^{th}$ row is used to measure the
distance to the affine function $l_{\vectorize{i}}^{'}$.
\end{proposition}
\begin{proof}
Let the distance of a Boolean function to a nonzero linear function be
expressed as
\[
d(f,l_i) = 2^{n-1}+\alpha
\]

where $\alpha$ is the sum of the distance coefficients except the constant coefficient. The value of $\alpha$ also corresponds to the sum of certain entries of the $i^{th}$ row of the LDM,
with each entry being multiplied with a constant depending on the Boolean function $f$, which will be explained in the next subsection. Regardless of this multiplication,
if the distance coefficients in the $i^{th}$ row of the LDM are negated, the new sum becomes
$2^{n-1}-\alpha$, and this is equivalent to $d(f,l_i\oplus 1)$, since $d(f,l_i\oplus 1)= 2^n - d(f,l_i)$.

\end{proof}

In view of Proposition \ref{prop_ldmpositive}, all entries of the LDM
can be considered as absolute values. This results in computing the
distance to the linear function $l_w$ if $wt(w)$ is even and to the
affine function $l_w^{'}$ if $wt(w)$ is odd.

\subsection{Combining Coefficients}
A Boolean function consisting of $p$ monomials has $2^p$ distance
coefficients associated with the nonzero products of ANF coefficients, which can be computed according to (\ref{eqn_lambda}). Computing the weight of the function requires all these coefficients to be added whereas the distance to a particular linear function can
be obtained by adding a subset of these coefficients plus a constant value of $2^{n-1}$. Computing the distance coefficients by processing all combinations of $p$
monomials has a computational complexity of $\mathcal{O}(2^p)$
operations, and this can be done at most for $p < 40$ in practice.
Now, a new method will be introduced which combines the related distance
coefficients, with the aim of reducing the number coefficients and avoiding the processing of  $2^p$ monomial combinations.

Let $a^{I}$ and $a^{J}$ be two terms in the distance function such
that $\bigcup\limits_{i\in I}^{}{supp(\vectorize{i})}=\bigcup\limits_{j\in
J}^{}{supp(\vectorize{j})}$, that is, the input variables appearing
in ANF coefficients of $a^I$ are the same as input variables
appearing in ANF coefficients of $a^J$. This not only makes
$n_I=n_J$ according to (\ref{eqn_ni}), but also specifies that these
terms appear in exactly the same truth table entries according to Proposition \ref{prop_union}. 
Hence, in the distance function \eqref{anf2nl_eqn_weight2}, for all values of $w$,
the distance coefficients of $a^I$ and $a^J$ will behave the same with respect to Proposition \ref{prop_lambdaw}. The distance coefficients of these terms can be collected under the distance coefficient of the term $a^K=a_k$ such that $a_k$ is the representative ANF coefficient of both $a^I$ and $a^J$. Since $n_I=n_J=n_K$, it is  $|I|$ and $|J|$ that distinguishes the distance coefficients $\lambda_I$
and $\lambda_J$ from $\lambda_K$. From (\ref{eqn_di}), it follows
that $\lambda_I=(-2)^{|I|-1}\lambda_K$ and
$\lambda_J=(-2)^{|J|-1}\lambda_K$. The distance coefficients of the terms which have the same representative ANF coefficients can be collected under a single distance coefficient, called the \textit{representative distance coefficient}, that is, the distance coefficient of the term belonging to the representative ANF coefficient. When this is done, the number of times a representative distance coefficient should be added will be called the \textit{combined coefficient}, and when multiplied with the value of the representative distance coefficient, will be called the \textit{combined distance coefficient}.
The combined coefficients of the distance function can be computed from the ANF coefficients of a Boolean function as follows:
\begin{equation}
\label{eqn_nnfcombined} C_{u}=\sum\limits_{\vectorize{u}=\vectorize{u_1}\vee\cdots\vee
\vectorize{u_k}}{(-2)^{k-1}\prod\limits_{1\leq i \leq k}{a_{u_i}}}.
\end{equation}

\begin{example}
Let the support of the ANF coefficients for a Boolean function be $\{a_1,a_2,a_3\}$. Then the distance coefficients corresponding to the nonzero products 
\[
\{1,a_1,a_2,a_3,a_1a_2,a_1a_3,a_2a_3,a_1a_2a_3\}
\]
 are
\[
\{\lambda, \lambda_1, \lambda_2,\lambda_3,\lambda_{1,2},\lambda_{1,3},\lambda_{2,3},\lambda_{1,2,3}\}
\]
where $\lambda$ is the constant term, $\lambda_1$ is the distance coefficient of $a_1$, $\lambda_{1,2}$
is the distance coefficient of $a_1a_2$, and so on. Since $\lambda_{1,2}=\lambda_{1,3}=\lambda_{2,3}=-2\lambda_3$ and $\lambda_{1,2,3}=4\lambda_3$ from (\ref{eqn_di}) and
(\ref{eqn_ni}), these distance coefficients can be combined under the representative distance coefficient $\lambda_{3}$, producing the combined distance coefficients $\{\lambda, \lambda_1,
\lambda_2, -\lambda_3\}$. Here, the corresponding combined coefficients are $\{1, 1, 1, -1\}$.

\end{example}

Algorithm \ref{alg_distancecoef} calculates the combined
coefficients from the ANF by processing each
monomial one by one. The input to the algorithm is a list
of monomials of the Boolean function called \texttt{MonList}. The monomials are represented by $x^I=\prod\limits_{i\in I}{x_i}$ where $I\subseteq\{1,\cdots,n\}$.
The output of the algorithm is a list of combined coefficients called \texttt{CoefList} consisting of
elements of the form $C_ix^I$ where $C_i\in\mathrm{Z}$ is the combined coefficient of the monomial $x^I$. When a new monomial $x^I$ is added to the function, the product of each existing combined coefficients and the new monomial is processed and the newly produced coefficients are added to a temporary list named \texttt{NewList}. For an existing coefficient $C_jx^J$, if it is the case that $I\subseteq J$, then the product of
these two coefficients will be $-2C_jx^J$, when added will negate the original coefficient. The products of terms whose $n_I$ part as specified in (\ref{eqn_ni}) will reside in the new monomial $x^I$ are collected under the variable \texttt{S}. These
terms will be added to the combined coefficients at the end of processing that monomial. The last case is the general case
where the coefficient of the term produced by the product of two monomials are added to the \texttt{NewList}. At the end of processing
the combined coefficients of the previous step, all the newly produced coefficients are added to \texttt{CoefList}. Addition of items to lists is denoted
by + operator. When an entry $C_ix^I$ is to be added to a coefficient list, if there already exists an element $C_jx^I$, then
the coefficient $C_ix^I$ is updated as $(C_i+C_j)x^I$.

\begin{algorithm}
\caption{Calculate combined coefficients}\label{alg_distancecoef}
\begin{algorithmic}[1]
\Procedure{CalcCoef}{$MonList$} 

\State $CoefList\gets\emptyset$

\ForAll{$x^I\in MonList$}
\State $NewList \gets \emptyset$
\State $S \gets 1$

\ForAll{$C_jx^J \in CoefList$}
\If{$I\subseteq J$}
\State $C_j \gets -C_j$
\ElsIf{$J\subset I$}
\State $S \gets S -2C_j$
\Else
\State $NewList \gets NewList + (-2C_jx^{I\cup J})$
\EndIf
\EndFor

\

\ForAll{$C_kx^K\in NewList$}
\State $CoefList\gets CoefList + C_kx^K$
\EndFor

\

\State $CoefList\gets CoefList +Sx^I$
 \EndFor

\

\State \textbf{return} $CoefList$
\EndProcedure
\end{algorithmic}
\end{algorithm}

The sum of all combined distance coefficients will give the weight
of the function. In order to compute the distance to a nonzero
linear function $l_w$, only a subset of these coefficients need to
be added, which is determined according to whether the vector associated with
the combined distance coefficient is contained in $\supvec{w}$.
Also, since $\lambda_I^w=2^{n-1}$ for $w\neq\vectorize{0}$ and $I=\emptyset$, a constant
value of $2^{n-1}$ should be added.

Since the distance to linear functions is related to the Walsh
coefficients of a Boolean function, the sum of the combined distance
coefficients can be expressed in terms of Walsh coefficients. For
$w\neq\vectorize{0}$, i.e., for a nonzero linear function $l_w$,
\begin{eqnarray}
\label{eqn_ldm_walsh} W_f(w)=2^n-2d(f,l_w),\\
W_f(w) = 2^n - 2(2^{n-1}+\alpha_w),\\
\label{eqn_alphaw}\alpha_w=-\frac{W_f(w)}{2}
\end{eqnarray}
where $\alpha_w$ is the sum of the combined distance coefficients, excluding the constant coefficient $2^{n-1}$. From \eqref{eqn_alphaw}, it can be seen that if the constant term $2^{n-1}$ is not added, the distance function \eqref{anf2nl_eqn_weight2} outputs $-\frac{W_f(w)}{2}$. Because the nonlinearity of a Boolean function depends on the maximum absolute value of the Walsh spectrum, being able to compute the maximum absolute value of the sum of the distance coefficients without adding the constant term is also sufficient to find out the nonlinearity.

\section{Computing Nonlinearity}

The nonlinearity computation algorithm consists of two phases. In the first phase, combined distance coefficients are calculated according to Algorithm \ref{alg_distancecoef} from the given ANF coefficients. Once this phase is completed, the distance
to any linear function can be obtained by adding a subset of these coefficients,
i.e., by adding the coefficients corresponding to columns $\supvec{w}$ in the $w^{th}$ row if the distance to the linear function $l_w$ is to be calculated. The nonlinearity computation on the other hand, requires all the distances to the linear functions to be calculated and the one with the minimum value being identified, which cannot be done in practice if $n$ is too high.
After the combined distance coefficients for a Boolean function are calculated, the distance to any linear function can be represented of the form

\begin{equation}
\label{eqn_bip}
F(b_1,\cdots,b_k)=\sum\limits_{1\leq i \leq k}^{}{\beta_i b_i}
\end{equation}
where $b_i\in\mathbb{F}_2$ and $\beta_i\in\mathrm{Z}$ is the combined distance coefficient associated with $b_i$.
Each $b_i$ in this function determines whether a zero or a nonzero entry is chosen from the corresponding column of the LDM.
Although there are $2^k$ possible inputs to this function, only some of the $k$-bit inputs actually correspond to a distance to a linear function. By enumerating all such $k$-bit inputs, one obtains the set of all distinct distances. The task of computing the nonlinearity then corresponds to finding the minimum of these values. Note, however, as explained in the previous section, omitting the addition of the constant coefficient when calculating the distance to the nonzero linear functions, one gets the negative half value of the Walsh coefficient, and this constant coefficient is assumed to be excluded in \eqref{eqn_bip}. With this slight modification, it becomes the maximum absolute value of \eqref{eqn_bip} to be found, instead of the minimum and maximum values, had the constant coefficient been added.

In the second phase of the nonlinearity computation algorithm, the maximum absolute value of the distance function \eqref{eqn_bip} is searched, which determines the nonlinearity. This problem also corresponds to a binary integer programming problem. The set of $2^k$ possible $k$-bit input vectors will be classified as feasible or infeasible according to whether they represent a distance to a linear function or not. The feasibility checking of inputs can be performed with Algorithm \ref{alg_branch1} and Algorithm \ref{alg_branch0} which determine whether a particular zero/nonzero choice of values in certain columns is possible in any of the rows of the LDM. By enumerating all possible distances to the set of linear functions, the minimum of these can be taken as the nonlinearity of the Boolean function.

A common approach in solving integer programming problems is to utilize the tree structure. The set of feasible inputs to the function (\ref{eqn_bip}) can be shown in a tree with the input variables $b_i$ being the nodes. This tree will be called the \emph{distance tree}. Starting with $b_1$ as the root node, each left (resp. right)
child node of a node $b_i$ represents the case where $b_i=1$ (resp. $b_i=0$). In order to enumerate feasible inputs, it is sufficient to check whether a node $b_i$ can take on 
the value 0 or 1 depending on the values of the parent nodes (values of the preceding variables). This can be done efficiently by using the facts mentioned in
Remarks (\ref{rmk_ldm4}), (\ref{rmk_ldm3}) and (\ref{rmk_ldm5}). Row indices $r$ in the LDM containing a nonzero entry in the $j^{th}$ column satisfy $\vectorize{r}\preceq\vectorize{j}$. Similarly, if a row has a zero value in the $j^{th}$ column then $supp(\vectorize{r}) \cap (\{1,\cdots,n\}\setminus supp(\vectorize{j})) \neq \emptyset$ must be satisfied
since a zero value in the $j^{th}$ column appears only in the row indices where at least one bit is set that is not in $supp(\vectorize{j})$.
Given two lists of columns $C_0$ and $C_1$, and another column identified by the index $u$,  Algorithm \ref{alg_branch1} determines whether there exists a row in the LDM containing a nonzero value in column $u$, with the condition that the
values in columns $C_0$ are zero and the values in columns $C_1$ are nonzero. Algorithm \ref{alg_branch0} performs the same task by checking
whether there is a row having a zero value in column $u$, under the same conditions. These two algorithms allow one to enumerate all feasible inputs to (\ref{eqn_bip}) by using the
associated column indices for each input variable $b_i$.

\begin{algorithm}
\caption{Decide whether a node in the distance tree can take on the value 1.}\label{alg_branch1}
\begin{algorithmic}[1]
\Procedure{Branch1}{$C_0,C_1,u$} 
\State $IncludeMask=\bigwedge\limits_{a\in C_1}{\vectorize{a}}$
\If{$(u\wedge IncludeMask)=\vectorize{0}$}
\State \textbf{return} false
\Else
\For{$i\in C_0$}
\If{$(\neg\vectorize{i} \wedge u \wedge IncludeMask)=\vectorize{0}$}
\State \textbf{return} false
\EndIf
\EndFor
\EndIf
\State \textbf{return} true
\EndProcedure
\end{algorithmic}
\end{algorithm}

\begin{algorithm}
\caption{Decide whether a node in the distance tree can take on the value 0.}\label{alg_branch0}
\begin{algorithmic}[1]
\Procedure{Branch0}{$C_0,C_1,u$} 
\State $IncludeMask=\bigwedge\limits_{a\in C_1}{\vectorize{a}}$
\If{$(\neg u\wedge IncludeMask)=\vectorize{0}$}
\State \textbf{return} false
\Else
\For{$i\in C_0$}
\If{$(\neg \vectorize{i} \wedge IncludeMask)=\vectorize{0}$}
\State \textbf{return} false
\EndIf
\EndFor
\EndIf
\State \textbf{return} true
\EndProcedure
\end{algorithmic}
\end{algorithm}

\textbf{Example.} Let $f(x_1,\cdots,x_5)=x_1x_5\oplus x_4x_5\oplus x_1x_2x_3\oplus x_1x_2x_4\oplus x_1x_2x_3x_4x_5$. Using Algorithm \ref{alg_distancecoef}, the following set of combined distance coefficients are obtained:
\[\{\lambda_3,\lambda_{17},\lambda_{26},\lambda_{28},-2\lambda_{19},-2\lambda_{29},-2\lambda_{30},3\lambda_{31}\}.\]
The associated distance function formed by these coefficients is
\begin{equation*}
F(b_1,\cdots,b_8)=8b_1+8b_2+4b_3+4b_4-8b_5-4b_6-4b_7+3b_8.
\end{equation*}
Note that in the LDM of order 5, $b_1$ is associated with column $a_3$, $b_2$ is associated with column $a_{17}$, and so on. Among the $2^8=256$ possible inputs to $F$, only $13$ of them are found to be feasible as shown in the distance tree in Figure \ref{fig_tree}. In the figure, the values in the leaf nodes shown in boxes below denote the output of $F$ when that particular combination of $b_i$'s are chosen. Each leaf node also corresponds to a feasible input of $F$ and is called a \textit{path}, denoted with a sequence of input bits. The value of a path is the output of $F$ for that input. A path in the tree identifies the linear functions whose distances to the Boolean function in consideration are obtained by adding the same combined distance coefficients, also specifying the distance to these functions. As there can be more than one linear function covered by a path, different paths can have the same value. For example, there are six paths with value 3, four paths with value -1 and two paths with value -5 in the distance tree of the example function. Table \ref{tab_tree} gives a more detailed information about this distance tree. The first column of the table denotes the leaf node (or path) number, the second column lists the path string, the third column gives the output of the distance function $F$ for the corresponding path and the last column denotes which linear functions that path is associated with. 

\begin{figure*}[!t]
\centering
\normalsize
\includegraphics[scale=0.75]{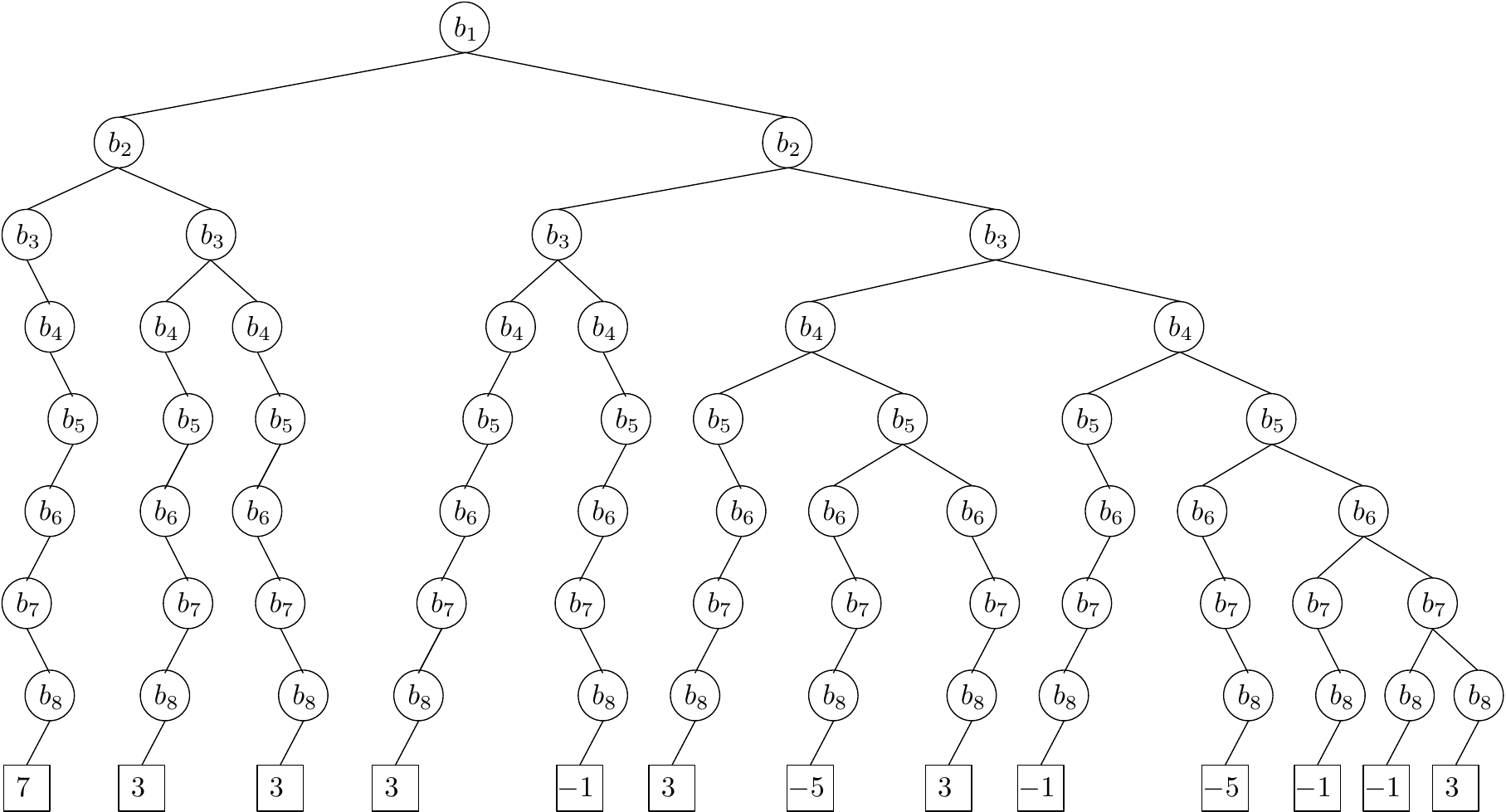}
\caption{Distance tree of $F(b_1,\cdots,b_8)=8b_1+8b_2+4b_3+4b_4-8b_5-4b_6-4b_7+3b_8$.}
\label{fig_tree}
\end{figure*}

\begin{table}[!t]
\renewcommand{\arraystretch}{1.35}
\caption{Distance tree data of $F(b_1,\cdots,b_8)=8b_1+8b_2+4b_3+4b_4-8b_5-4b_6-4b_7+3b_8$.}
\label{tab_tree} \centering
\begin{tabular}{c |c |c| r}
\hline N & Path & F & Associated $l_w$\\ \hline
1 & 11001101 & 7 & $x_5$\\\hline
2 & 10101011 & 3 & $x_4$\\\hline
3 & 10001001 & 3 & $x_4\oplus x_5$\\\hline
4 & 01111111 & 3 & $x_1$\\\hline
5 & 01001101 & -1& $x_1\oplus x_5$\\\hline
6 & 00110111 & 3 & $x_2$\\
  &          &   & $x_1\oplus x_2$\\\hline
7 & 00101011 & -5 & $x_1\oplus x_4$\\\hline
8 & 00100011 & 3 & $x_2\oplus x_4$\\
  &          &   & $x_1\oplus x_2\oplus x_4$\\\hline
9 & 00010111 & -1 & $x_3$\\
  &          &    & $x_1\oplus x_3$\\
  &          &    & $x_2\oplus x_3$\\
  &          &    & $x_1\oplus x_2\oplus x_3$\\\hline
10& 00001001 & -5 & $x_1\oplus x_4\oplus x_5$\\\hline
11& 00000101 & -1 & $x_2\oplus x_5$\\
  &          &    & $x_1\oplus x_2\oplus x_5$\\
  &          &    & $x_3\oplus x_5$\\
  &          &    & $x_1\oplus x_3\oplus x_5$\\
  &          &    & $x_2\oplus x_3\oplus x_5$\\
  &          &    & $x_1\oplus x_2\oplus x_3\oplus x_5$\\\hline
12& 00000011 & -1 & $x_3\oplus x_4$\\
  &          &    & $x_1\oplus x_3\oplus x_4$\\
  &          &    & $x_2\oplus x_3\oplus x_4$\\
  &          &    & $x_1\oplus x_2\oplus x_3\oplus x_4$\\\hline
13& 00000001 & 3  & $x_2\oplus x_4\oplus x_5$\\
  &          &    & $x_1\oplus x_2\oplus x_4\oplus x_5$\\
  &          &    & $x_3\oplus x_4\oplus x_5$\\
  &          &    & $x_1\oplus x_3\oplus x_4\oplus x_5$\\
  &          &    & $x_2\oplus x_3\oplus x_4\oplus x_5$\\
  &          &    & $x_1\oplus x_2\oplus x_3\oplus x_4\oplus x_5$\\
\hline
\end{tabular}
\end{table}

\begin{table}[!t]
\renewcommand{\arraystretch}{1.3}
\caption{Distance coefficients and related portion of the LDM for $f(x_1,\cdots,x_5)=x_1x_5\oplus x_4x_5\oplus x_1x_2x_3\oplus x_1x_2x_4\oplus x_1x_2x_3x_4x_5$.} \label{tab_example}
\centering
\begin{scriptsize}
\begin{tabular}{l|r|r|r|r|r|r|r|r||r|r|r}
\hline $C_i$& 1&1&1&1&-2&-2&-2&3&F&W& N\\ \hline 
$\beta_i$& \bfseries 8& \bfseries 8& \bfseries 4& \bfseries 4& \bfseries -8 & \bfseries -4& \bfseries -4& \bfseries 3& &\\ \hline
 & $a_{3}$ & $a_{17}$ & $a_{26}$ & $a_{28}$ & $a_{19}$ &$a_{29}$ & $a_{30}$ & $a_{31}$& & \\\hline \hline 
$l_{0}$  &  8   &   8   &   4   &   4   &   4   &   2   &   2   &   1  &   11   & 10&-\\ \hline 
$l_{1}$  & -8   &   -8  &   0   &   0   &   -4  &   -2  &   0   &  -1  & 7  & 14&1\\ \hline 
$l_{2}$  & -8   &   0   &   -4  &   0   &   -4  & 0 & -2 &         -1  &   3  & 6&2\\ \hline 
$l_{3}$  & 8    &   0   &   0   &   0   & 4 & 0 & 0   &             1   &   3   & -6&3\\ \hline 
$l_{4}$  & 0    &   0   &   0   & -4 & 0 & -2 &   -2  &            -1  &   -1   & -2&9\\ \hline 
$l_{5}$  & 0   &   0   &   0   & 0 &0 & 2 &   0   &                 1   &    -1  & 2&11\\ \hline 
$l_{6}$  & 0   &   0   & 0   & 0 & 0 &0 & 2   &                     1   &   -1  & 2&12\\ \hline 
$l_{7}$  & 0   & 0   & 0   & 0   & 0 & 0 & 0&                      -1  &   3  & 6&13\\ \hline 
$l_{8}$  & 0   &0   &   -4  &   -4  &   0   &   -2  &   -2  &      -1  &   3   & 6&6\\ \hline
$l_{9}$  & 0   &   0   &   0   &   0   &   0   &   2   &   0    &   1   & -1  & 2&11\\ \hline 
$l_{10}$ & 0   &   0   &   4   &   0   &   0   &   0   &   2      & 1   &   3   & -6&8\\ \hline 
$l_{11}$ & 0   &   0   &   0   &   0   &   0   &   0   & 0   &     -1  &   3  &6 &13\\ \hline 
$l_{12}$ & 0   &   0   &   0   &   4   &   0   & 2   &   2   &      1   &   -1  & 2&9\\ \hline 
$l_{13}$ & 0   &   0   &   0   &   0   & 0   &   -2  &   0   &     -1  &   -1   & -2&11\\ \hline 
$l_{14}$ & 0   &   0   &   0   & 0   &   0   &   0   &   -2  &     -1  &   -1   & -2&12\\ \hline 
$l_{15}$ & 0   &   0   & 0   &   0   &   0   &   0   &   0   &      1   &   3   & -6&13\\ \hline 
$l_{16}$ & 0   & -8  &   -4  &   -4  &   -4  &   -2  &   -2  &     -1  &   3   & 6&4\\ \hline 
$l_{17}$ & 0   &   8   &   0   &   0   &   4   &   2   &   0   &    1   &-1 &  2&5\\ \hline 
$l_{18}$ & 0   &   0   &   4   &   0   &   4   &   0   &   2   &    1   &   -5 & 10 &7\\ \hline 
$l_{19}$ & 0   &   0   &   0   &   0   &   -4  &   0   &0   &      -1  &   -5  &  -10&10\\ \hline 
$l_{20}$ & 0   &   0   &   0   &   4   &   0   & 2   &   2   &      1   &   -1 & 2 &9\\ \hline 
$l_{21}$ & 0   &   0   &   0   &   0   & 0   &   -2  &   0   &     -1  &   -1  &  -2&11\\ \hline 
$l_{22}$ & 0   &   0   &   0   & 0   &   0   &   0   &   -2  &     -1  &   -1  &  -2&12\\ \hline 
$l_{23}$ & 0   &   0   &0   &   0   &   0   &   0   &   0   &       1   &   3 &  -6&13\\ \hline 
$l_{24}$ & 0   & 0   &   4   &   4   &   0   &   2   &   2   &      1   &   3 &  -6&6\\\hline 
$l_{25}$ & 0   &   0   &   0   &   0   &   0   &   -2  &   0   &   -1  &-1  &  -2&11\\ \hline 
$l_{26}$ & 0   &   0   &   -4  &   0   &   0   &   0   &   -2  &   -1  &   3  &6 &8\\ \hline 
$l_{27}$ & 0   &   0   &   0   &   0   &   0   &   0   &0   &       1   &   3   &-6 &13\\ \hline 
$l_{28}$ & 0   &   0   &   0   &   -4  &   0   &-2  &   -2  &      -1  &   -1   & -2&9\\ \hline 
$l_{29}$ & 0   &   0   &   0   &   0   &0   &   2   &   0   &       1   &   -1  & 2&11\\ \hline 
$l_{30}$ & 0   &   0   &   0   &0   &   0   &   0   &   2   &       1   &   -1  & 2&12\\ \hline 
$l_{31}$ & 0   &   0   &0   &   0   &   0   &   0   &   0   &      -1  &   3  & 6&13\\ \hline
\end{tabular}
\end{scriptsize}
\end{table}

Table \ref{tab_example} shows a portion of the LDM of order 5, where only the eight columns related to the distance function of the example are shown.
The distance tree constructed using Algorithm \ref{alg_branch1} and Algorithm \ref{alg_branch0} actually enumerates the distinct sums that can occur when the combined distance coefficients are added for each row of the LDM.
In Table \ref{tab_example}, the top most row denotes the combined coefficients $C_i$ of the example function and the row below denotes the combined distance coefficients $\beta_i$, obtained by multiplying the $C_i$'s with the representative distance coefficients (the positive value appearing in the associated column).
The values in this row constitute the coefficients of the distance function $F$ whose absolute maximum value is to be searched. The right most three columns of the table list the output values of the distance function $F$,
the Walsh coefficients of $f$ and which leaf node in the distance tree this row belongs, among the nodes listed in Table \ref{tab_tree}. Note that in Table \ref{tab_example}, for the
rows $l_i$ with $wt(\vectorize{i})$ being odd, that is, the rows corresponding to the linear functions with odd number of terms, $F$ outputs the negative of what must
actually be computed. This is because all values in the LDM are taken to be positive as a consequence of Proposition \ref{prop_ldmpositive}. The values listed in column $F$
of Table \ref{tab_example} correspond to $-\frac{W_f(\vectorize{i})}{2}$ for rows $l_i$ if $wt(\vectorize{i})$ is even, and $\frac{W_f(\vectorize{i})}{2}$ if $wt(\vectorize{i})$ is odd. This means that the absolute maximum value of the values produced in the distance tree can be used to find out the nonlinearity. The maximum value appearing in the distance tree is 7 which
appears in the left most leaf node. Therefore, the nonlinearity is $2^{n-1}-\frac{1}{2}\max\limits_{w\in \mathbb{F}_2^n}{|W_f(w)|}$,
which is $16 - 7=9$. It must also be noted that the first row of the LDM which is used to compute the weight is not taken into account here. Since the weight of the Boolean function can obtained by adding the combined distance coefficients produced in the first phase of the nonlinearity computation algorithm, 
it is sufficient to check whether this value is smaller than the nonlinearity value obtained in the second phase or not. The weight of the example function is $F(1,\cdots,1)=11$, which is larger than 9, so the nonlinearity is found to be 9. A better use of the weight computed in the first phase is to set it as the best solution and make use of the optimization techniques in the second phase for solving the integer programming problem more efficiently.

\subsection{Branch and Bound Method}

Branch and bound method is a way of efficiently solving integer programming problems by early terminating the processing of nodes if the best solution that can be obtained from a node will not attain the maximum/minimum value whatever the subsequent variable choices are \cite{Ch00}. This idea could be realized in this specific problem by computing one maximum and one minimum value for each variable in the optimization problem. Since the variables are processed one by one, this allows one to determine what the maximum and minimum change in the function could be, regardless of whether the assignment of inputs are feasible or not. The algorithm then can choose not to branch a node if it is guaranteed that neither of the solutions obtained from that node will be better than the best feasible solution already obtained.
For the distance function described in (\ref{eqn_bip}) having $k$ input bits, the maximum amount of change (both positive and negative) for each node can be computed as follows:
\begin{eqnarray}
max_i&=&\sum\limits_{i\leq x\leq k,\, \beta_x > 0}{\beta_x}\\
min_i&=&\sum\limits_{i\leq x\leq k,\, \beta_x < 0}{\beta_x}.
\end{eqnarray}
Values $max_i$ and $min_i$ specify how much the function can increase and decrease at most, once the first $i-1$ variables are fixed.
For instance, the function can increase the most if all of the subsequent coefficients with $\beta_x > 0$ are added and $\beta_x < 0$ are omitted. Hence,
at a particular node while constructing the distance tree, by using the $max_i$ and $min_i$ values, it can be checked whether 
the processed node can yield a larger absolute value than the best one at hand. If not, the processing of that branch of the tree is terminated at that point. If the
node promises to attain a better solution, the branching continues. However, this does not guarantee that a better solution will be obtained
since an increase (resp. decrease) of $max_i$ (resp. $min_i$) might not be possible if these inputs are not feasible.

\subsection{Recovering the Nearest Affine Function}
Besides computing the nonlinearity, it could be as much important to identify which affine function(s) a Boolean function is closest to. This can be accomplished by using the path string of the nonlinearity algorithm described above.
The path string identifies the rows in a LDM by specifying certain columns containing either zero or nonzero entries.
This problem can be rephrased as follows: 

Given $n$ and two sets of indices $I=\{i_1,\ldots,i_k\}$, $E=\{e_1,\ldots,e_l\}$ 
where the indices are from the set $\{0, \ldots ,2^n-1\}$. Identify the rows $r$ in $M^n$ satisfying $M_{r,j}^n\neq 0$ for $j\in I$ and 
$M_{r,j}^n=0$ for $j\in E$. 

Algorithm \ref{alg_nearestaffine} outputs a list of linear functions identified by the vector $x\in\mathbb{F}_2^n$, given the column index sets $I$ and $E$. The algorithm makes use of Remark \ref{rmk_ldm4} and \ref{rmk_ldm5} to list the rows of the LDM satisfying the given conditions. This algorithm can be executed for each leaf node of the distance tree to find out the
linear functions at a specified distance to the Boolean function in question. Linear functions listed in the last column of Table \ref{tab_tree}
are found by executing Algorithm \ref{alg_nearestaffine} with the column index sets $I$ and $E$ constructed from the corresponding path strings. For example, in the first row of the table, the path string is $11001101$, which makes $I=\{3,17,19,29,31\}$ and $E=\{26,28,30\}$.
There is only one row in the LDM whose column indices specified in $I$ are nonzero and column indices specified in $E$ are zero, and that row corresponds to the linear function $l_1=x_5$.
The decision of whether the Boolean function is closer to a linear function or its complement can be made based on the sign of the Walsh coefficient.

\begin{algorithm}
\caption{Enumerate the linear functions associated with a path in the distance tree}\label{alg_nearestaffine}
\begin{algorithmic}[1]
\Procedure{PathToLinearFunction}{$n,I,E$} 
\State $IncludeMask=\bigwedge\limits_{i\in I}{\vectorize{i}}$
\ForAll{$x\in \subvec{IncludeMask}$}
\State valid=\textbf{true}
\ForAll{$y\in E$}
\If{$(\neg y \wedge x) = \vectorize{0}$}
\State valid = \textbf{false}
\State break
\EndIf
\EndFor

\If{valid=\textbf{true}}
\State print $x$
\EndIf
\EndFor
\EndProcedure
\end{algorithmic}
\end{algorithm}

\subsection{Complexity of the Algorithm and Experimental Results}
The nonlinearity computation algorithm consists of two phases. In the first phase, from a given set of ANF coefficients, the combined coefficients are calculated. The number of combined coefficients is equal to the number of distinct products of input monomials. Although this value can be as high as $2^p$ ($p$ being the number of monomials) where each monomial combination is distinct, the actual value will vary depending on the structural relations between monomials. The expected value of this quantity is given in \cite{CD12} in terms of $n$ and $p$ as

\begin{equation*}\label{eqn_expected_s1b}
\sum\limits_{k=1}^{p}{(1-(1-(1-q)^k)^n) \binom{n}{k}}\\
\end{equation*}
where $q$ denotes the probability of a variable appearing in a monomial. For randomly generated Boolean functions, $q$ is $\frac{1}{2}$.

For the second phase of the algorithm where the binary integer programming problem is solved,
it is harder to give an explicit expression of the complexity. The difficulty of estimating the complexity is a result of the fact that it depends on the distribution of the values in the Walsh spectrum of the Boolean function. However, the expreriments indicate that the execution time of this phase is negligible compared to the first phase of the algorithm.

The proposed nonlinearity computation algorithm is implemented in C language and execution times
are measured for different parameters on a PC having an Intel Core2 Duo processor running at $3.0\mathrm{GHz}$.
Table \ref{tab_timing} gives the average running times for 60-variable Boolean functions with branch and bound method being employed. In the table, $p$ denotes the number of monomials, $k$ is the average number of combined distance coefficients, i.e., the number of variables
of the associated integer programming problem and the last column denotes the average running times of the algorithm. For each number of monomials in the experiments, average timings were calculated over 10 randomly generated Boolean functions. 
\begin{table}[!h]
\renewcommand{\arraystretch}{1.3}
\caption{Timings for $n=60$.}
\label{tab_timing} \centering
\begin{tabular}{c |r r r}
\hline p & k & Time (sec.)\\ \hline
30 &20603&1 \\
40 &54742&11 \\
50 &112823&49 \\
60 &215641&261 \\
70 &354954&973 \\
80 &584629&3191 \\
90 &833474&6878 \\
100&1176612&14938\\
\end{tabular}
\end{table}

\section{Conclusion}
An algorithm for computing the nonlinearity of a
Boolean function from its ANF coefficients is proposed. The
algorithm makes use of the formulation of the distance of a Boolean
function to the set of linear functions. It is shown that the
problem of computing the nonlinearity corresponds to a binary integer programming problem where techniques for efficiently solving these problems such as branch and bound method can be applied to improve the performance. The algorithm allows the computation of nonlinearity
for Boolean functions acting on large number of inputs where applying
the Fast Walsh transform is impractical.

\ifCLASSOPTIONcaptionsoff
  \newpage
\fi


\begin{thebibliography}{1}

\bibitem{MS77}
F.~J. MacWilliams, and N.~J.~A. Sloane, \emph{The Theory of Error
Correcting Codes}, North-Holland, 1977.

\bibitem{GS09}
K.~Chand Gupta and P.~Sarkar, \emph{Computing Partial Walsh Transform From the Algebraic Normal Form of a Boolean Function}, IEEE Transactions on Information Theory 55(3), 1354-1359 (2009)

\bibitem{Ca13}
\c{C}.~\c{C}al\i k, \emph{Computing Cryptographic Properties of Boolean Functions from the Algebraic Normal Form Representation}, Ph.D. Thesis, Middle East Technical University, Feb. 2013.

\bibitem{CD12}
\c{C}.~\c{C}al\i k and A.~Do\u{g}anaksoy, \emph{Computing the Weight of a Boolean Function from its Algebraic Normal Form}, SETA 2012, T.Helleseth and J.Jedwab (Eds.), LNCS 7280, pp. 89-100, Springer, 2012.

\bibitem{CG99}
C.~Carlet and P.~Guillot, \emph{A New Representation of Boolean Functions},
Proceedings of AAECC'13, Lecture Notes in Computer Science, 1719, pp. 99-103, 1999.

\bibitem{Ch00}
C.W.~Chinneck, \emph{Practical Optimization: A Gentle Introduction},
online textbook,
http://www.sce.carleton.ca/faculty/chinneck/po.html.


\end{thebibliography}
\end{document}